\documentclass[]{elsarticle}
\pdfoutput=1
\usepackage{graphicx}
\usepackage{amsmath}
\usepackage{geometry}
\usepackage{textcomp}
\usepackage{amstext}
\usepackage{amssymb}
\usepackage{pgfplots}
\usepackage{esint,natbib}
\usepackage{subfigure,multirow}
\usepackage{color,soul}
\usepackage{url}
\usepackage{textcomp}
\usepackage{booktabs}
\usepackage{bm}
\newtheorem{theorem}{Theorem}[section]
\newtheorem{lemma}[theorem]{Lemma}
\newenvironment{proof}[1][Proof]{\begin{trivlist}
\item[\hskip \labelsep {\bfseries #1}]}{\end{trivlist}}

    \usetikzlibrary{patterns}

    \makeatletter
\def\ps@pprintTitle{%
 \let\@oddhead\@empty
 \let\@evenhead\@empty
 \def\@oddfoot{}%
 \let\@evenfoot\@oddfoot}
\makeatother

\begin{document}
\begin{frontmatter}
\title{Credal Model Averaging for classification: representing  prior ignorance and  expert opinions.}
\author[idsia]{Giorgio Corani\corref{cor1}}
\ead{giorgio@idsia.ch}
\author[elet]{Andrea Mignatti}
\ead{andrea.mignatti@polimi.it}
\cortext[cor1]{Corresponding author}

\address[idsia]{Istituto Dalle Molle di Studi sull'Intelligenza Artificiale (IDSIA)\\Galleria 2, 6928 Manno (Lugano), Switzerland}
\address [elet]{Dipartimento di Elettronica, Informazione e Bioingegneria\\ Politecnico di Milano, Italy\\}

\begin{abstract}
Bayesian model averaging (BMA) is the state of the art approach for overcoming model uncertainty.
Yet, especially on small data sets, the results yielded by BMA might be sensitive to the prior over the models.
Credal Model Averaging (CMA) addresses this problem by substituting the single prior over the models by a \textit{set} of priors (credal set). 
Such approach solves the problem of how to choose the prior over the models and automates sensitivity analysis.
We discuss various CMA algorithms for building an ensemble of logistic regressors characterized by different sets of covariates. We show how CMA can be appropriately tuned to the case in which one
is prior-ignorant and to the case in which instead domain knowledge is available.
CMA detects \textit{prior-dependent} instances, namely instances in which a different class
is more probable depending on the prior over the models. On such instances CMA suspends the judgment, returning multiple classes.
We thoroughly compare different BMA and CMA variants on a real case study,
predicting presence of Alpine marmot burrows in an Alpine valley.
We find that BMA is almost a random guesser on the instances recognized as prior-dependent by CMA.
\end{abstract}
\end{frontmatter}

\section{Introduction}
\textit{Classification} is the problem of  predicting the outcome of a categorical variable
on the basis of several variables (called \textit{features} or \textit{covariates}).
However, there is often considerable uncertainty about which covariates should be included in the classifier.
Typically different sets of covariates are plausible given the available data. In this case  drawing
conclusions on the basis of the supposedly best single model can lead to overconfident conclusions,
overlooking the uncertainty of model selection (\emph{model uncertainty}).

Bayesian model averaging (BMA) \citep{hoeting1999bma} is a principled solution to model uncertainty.
BMA combines the inferences of multiple models;
the weights of the combination are the models' posterior probabilities.
However the results of BMA can be sensitive on 
the \textit{prior} probability assigned to the different models.
A common approach  is to assign equal prior probability
to all models (\textit{uniform prior}).
A more sophisticated solution is to adopt a \textit{ hierarchical} prior over the models, which yields inferences
less sensitive of the choice of the prior parameters \citep{Clyde2004,ley2009effect}.

However the specification of any prior implies some arbitrariness, which can lead to risky conclusions;
such risk is especially present on small data sets.
Often BMA studies \citep{wintle2003,link2006} report a sensitivity analysis, 
presenting the results obtained considering different priors over the models.

To robustly deal with the specification of the prior over the models,
we adopt a \textit{set} of priors (\textit{credal set}) over the models.
We thus adopt the paradigm of 
\emph{credal classifiers} \citep{zaffalon2001a} which extend traditional classifiers
by considering \textit{sets }of probability distributions. The main characteristic of credal
classifiers is that they allow for set-valued predictions of classes, when returning a single class is not deemed safe. Credal classifiers have been developed in the area of imprecise probability \citep{walley1991statistical}.

\textit{Credal model averaging} (CMA) \citep{corani2008credal} generalizes BMA by substituting the prior over the models by a credal set.
CMA  thus combines a set of traditional classifiers using imprecise probability.
CMA was firstly introduced \citep{corani2008credal} to create an imprecise ensemble of naive Bayes classifiers.
CMA adopts the credal set to express weak beliefs about the model prior probabilities:
by doing so, it does not commit to a single prior over the models.
As it is typical of credal classifiers,  CMA compute inferences which return \textit{interval} probabilities rather than single probabilities. For example, when classifying an instance
CMA computes the \textit{upper} and the \textit{lower} posterior probability of each class.
The length of the interval shows the sensitivity of the posterior on the prior over the models,
automating sensitivity analysis.
CMA  identifies \textit{prior-dependent} instances, namely instances in which a different class is more probable depending on the prior over the models.

In \citet{2013:23:isipta} we studied the problem of robustly predicting the
presence of Alpine marmot (\emph{Marmota marmota}) on the basis of several environmental
covariates  (slope, altitude, etc.).
Bayesian model averaging of logistic regressors is the state of the art approach
for analyzing presence/absence data \citep{wintle2003,link2006, thomson2007predicting}.
We thus devised \citep{2013:23:isipta} CMA for logistic regression
considering a constrained class of priors over the models which allowed for an analytical solution of the optimization problems. The credal set of CMA modeled a condition close to prior near-ignorance.
Moreover we presented some preliminary
results on the data set of presence of Alpine marmot collected by AM.
In particular we compared CMA against the BMA induced using the uniform prior over the models.

In this paper we extend in several respects our previous work.
From the algorithmic viewpoint we consider a more general class of distributions for the
prior probability of the models.
The new class of priors is a straightforward generalization of the previous one; yet it allows representing  prior knowledge in a much more flexible way.
As a side-effect, the  new class of priors requires a numerical solution of the
optimization problems.

We discuss three different CMA variants. The first
is our previous algorithm \citep{2013:23:isipta}.
The  new algorithm based on the more general class of priors yields two variants: one
referring to prior ignorance and one referring
to  partial prior knowledge.
To elicit prior knowledge we interviewed three experts:
two scientists who published several papers on the species
and a master student who participated in the collection of marmot data without analyzing them.

We present also a much extended empirical analysis of the Alpine marmot.
We consider the three mentioned CMA variants
and three BMA variants, which differ in the prior over the models.
Two priors are non-informative (uniform and hierarchical); the third prior is instead
based on the expert statements and is thus informative.

We assess not only the classification performance but also another important inference,
namely the posterior probability of inclusion of the covariates.

The paper is organized as follows: Section \ref{sec:bma} and  \ref{sec:cma-algo}
present the BMA and CMA algorithms; Section \ref{sec:case-study} describes the case study of
Alpine marmot and the interview of the experts;
Section \ref{sec:results} presents the empirical results.

\section{Logistic regression and Bayesian model averaging}\label{sec:bma}
The goal is to predict the outcome of the  \textit{binary} class variable $C$ which can assume values
$c_0$ or $c_1$.  
There are $k$   \textit{covariates} $\{X_1, X_2, \dots X_k\}$; an observation of the set of covariates is $\mathbf{x}=\{x_1,\ldots,x_k\}$.
Given $k$ covariates, $2^{k}$ different subsets of covariates can be defined; each subset of covariates yields 
a \textit{model structure} (or, more concisely, a \textit{structure}).
We denote by $m_i$ the i-\emph{th} model structure, by $\mathcal{X}_i$ its set of covariates and by $P(m_i|D)$ its posterior probability.
A training set of size $D$ is available for learning the models.
The data set has size $n$, namely it contains $n$ joint observations of the covariates and the class.
We denote as $P(c_1|D,\mathbf{x},m_i)$ the posterior probability of $c_1$  given
the covariate values  $\mathbf{x}$ and the 
model $m_i$  which has  been trained on data set $D$.
The logistic regression model is:
\begin{eqnarray}\label{eq:logistic}
\eta_{D,\mathbf{x},m_i}=\log\left(\frac{P(c_1|D,\mathbf{x},m_i)}{1-P(c_1|D,\mathbf{x},m_i)}\right)=
\beta_0 + \sum_{X_l \in \mathcal{X}_i}\beta_l x_{l} \label{eq:logit}
\end{eqnarray}
where $\eta_{D,\mathbf{x},m_i}$ denotes the \textit{logit} of the posterior probability of presence, $x_{l}$ the observation of  $l$-th covariate which has been included in model $m_i$ and $\beta_l$ its coefficient.

BMA addresses model uncertainty by combining the inferences of multiple models, and weighting them by the models' posterior probability.
The posterior probability of presence  is thus obtained by marginalizing out the model variable \citep{hoeting1999bma}:
\begin{eqnarray}
P(c_1|D,\mathbf{x})=\sum_{m_i\in\cal{M}} P(c_1|D,\mathbf{x},m_i) P(m_i|D)\label{eq:bma}
\end{eqnarray}
where $\cal{M}$ denotes the model space, which  contains the $2^k$ logistic regressors obtained considering all the possible subsets of features.
The posterior probability of $m_i$ given the data is computed as follows:
\begin{eqnarray} \label{eq:model_posterior}
 P(m_{i}|D)=\frac{P(m_{i})P(D|m_{i})}{\sum_{m_i\in \cal{M}}P(m_{i})P(D|m_{i})} 
\end{eqnarray}
where $P(m_{i})$ and $P(D|m_{i})$  are respectively the prior probability and the \textit{marginal likelihood} of model $m_i$.
The marginal likelihood integrates the likelihood with respect to the model parameters:
\[
P(D|m_{i})=\int P(D|m_{i},\bm\beta_i) P(\bm\beta_i|m_i) d\bm\beta_i
\]
where $\bm\beta_i$ denotes the set of parameters of model $m_i$.

A convenient approximation for computing the
models' marginal likelihood is based on the BIC \citep{raftery1995bayesian}.
The BIC of model $m_i$ is
\begin{equation}
 \mathrm{BIC_i}= -2 LL_i + |\bm{\beta}_i| \log (n)
\end{equation}
where $ LL_i$ denotes the log-likelihood of $m_i$, $|\bm{\beta}_i|$ the number of its parameters and $n$ the number
of data points on the data set.

The marginal likelihood of model $m_i$ can be approximated as:
\begin{equation}
 P(D|m_i)\approx\frac{\exp(-BIC_i/2)}{\sum_{m_i\in\cal{M}}\exp(-BIC_i/2)}.\label{eb:bic-lik}
\end{equation}
This approximation is convenient from a computational viewpoint and generally accurate; therefore, it is often adopted to compute BMA \citep{raftery1995bayesian,wintle2003,link2006}. Using the BIC approximation it is no longer necessary specifying the prior probability 
$P(\bm\beta_i|m_i)$ of the model parameters.


The posterior probability of model $m_i$ is then approximated as:
\begin{equation}
 P(m_i|D)\approx\frac{\exp(-BIC_i/2)P(m_i)}{\sum_{m_i\in\cal{M}}\exp(-BIC_i/2)P(m_i)}.
\end{equation}

A large number of covariates implies a huge model space, making it necessary to approximate the summation of Eqn.~(\ref{eq:bma});
computational strategies to this end are discussed for instance by \cite{Clyde2004}.
However our experiments involve a limited
number of covariates and thus we exhaustively sample the model space.

Often one is interested in the posterior probability of inclusion of feature $X_j$. This is the sum the posterior probabilities of the model structures
which do include $X_j$:
\begin{equation}
    P(\beta_j \neq 0)=\sum_{m_i \in \mathcal{M}}\rho_{ij}P(m_i|D)
\label{eq:post-inclusion}
    \end{equation}
where the binary variable $\rho_{ij}$ is 1 if model $m_i$ includes covariate $X_j$ and otherwise.

\subsection{Non-informative prior over  the models}\label{sec:priors}
A simple approach to set  the prior probability of the models is the \textit{independent Bernoulli} prior (IB prior).
The IB prior assumes that  each covariate is independently included in the model with identical probability $\theta$ \citep{Clyde2004}.
Denoting by $k_i$ the number of covariates included by model $m_i$ and by $k$ the total number of covariates,
the  prior probability of model $m_i$ is :
\begin{equation}\label{eq:binom}
 P(m_i)= \theta^{k_i} (1-\theta)^{k-k_i}
\end{equation}
which depends on the single parameter $\theta$.
By setting $\theta$=1/2 one obtains
the \textit{uniform} prior \textit{over the models}, which assigns to each model equal probability $1/2^k$.

However the uniform prior is quite informative if analyzed from the viewpoint of the 
\textit{model size}, namely the number of covariates included in the model.
We denote the model size by  $W$.
The IB prior implies $W$ to be binomially distributed: $W \sim Bin (\theta, k)$ \citep{ley2009effect}.
As well-known, the binomial distribution is far from flat.

A  flat prior distribution over the model size can be obtained by adopting the Beta-Binomial (BB) prior \citep{ley2009effect, Clyde2004}.
Compared  to the IB prior, the BB prior yields posterior inferences which are less sensitive on the value of $\theta$.
The BB has been recently recommended also for handling the problem of multiple hypothesis testing 
\citep{Scott2006,Carvalho2009}.
 
The BB prior treats the parameter $\theta$ as a random variable with Beta prior distribution:
$\theta \sim Beta(\alpha, \beta)$.
It is common to set  $\alpha=\beta=1$; under this choice,  the Beta distribution is  \textit{uniform}.
The resulting probability of model $m_i$ which contains $k_i$ covariates is \citep{ley2009effect}:
\begin{equation}
P(m_i)=\frac{k_i!(k-k_i)!}{k+1!} \label{eq:prob-beta-single}
%
\end{equation}
The resulting probability of the model size $W$ to be equal to $k_i$ is:
\begin{equation}
P(W=k_i)= 
\frac{1}{k+1}  \,\, \forall \,\, m_i \label{eq:w-k-i}
\end{equation}
The model size is thus uniformly distributed, as a result of having set a uniform prior on $\theta$.
In the Appendix we show the analytical derivation of formulas (\ref{eq:prob-beta-single})--(\ref{eq:w-k-i}).

Summing up, the IB prior under the choice $\theta=1/2$  implies all models to be equally probable and the model size to be  binomially distributed.
Instead the BB prior under the choice $\alpha=\beta=1$ implies the probability of each model to depend on the number of covariates according to Eqn.(\ref{eq:prob-beta-single}) and the model size $W$ to be uniformly distributed.

In Fig.\ref{bb-ib-prior} we compare the prior distribution on the model size $W$ obtained using the IB and the BB prior  for $k$=6; this is the number of covariates of our case study.

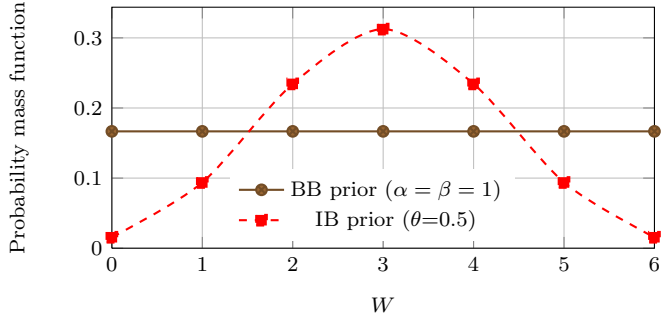
\begin{figure}[!htp]
 \centering
\begin{tikzpicture}[]
\begin{axis}[xlabel=$W$, ylabel=Probability mass function,
legend style={at={(axis cs:1.3,0.01)},anchor=south west,draw=none},
grid=major,  xtick={0,1,2,3,4,5,6},
ymin=0,xmin=0,xmax=6, height= 0.2 * \textheight, width= 0.5 * \textwidth,
every axis/.append style={font=\footnotesize}]
\addplot[smooth, thick,brown!60!black,every mark/.append style={fill=brown!80!black},mark=otimes* ] table [col sep=comma, x=W, y=beta] {beta-bin.csv};
\addplot [smooth,thick, dashed,red,mark=square*]
table [col sep=comma, x=W, y=bernoulli ] {beta-bin.csv};
\legend{BB prior ($\alpha=\beta=1$), IB prior ($\theta$=0.5)}
\end{axis}
\end{tikzpicture}
\caption{Prior distribution on the model size, under the independent Bernoulli and the beta-binomial prior for $k$=6.}\label{bb-ib-prior}
\end{figure}

\subsection{Informative prior}\label{sec:nb-prior}
One can express domain knowledge by differently specifying
the prior probability of inclusion of each covariate.
This requires generalizing the IB prior so that each covariate has its own
prior probability of inclusion.
We denote by $\boldsymbol{\theta}$
the [$k\times1$] vector including the prior probability of inclusion of covariates $X_1,\ldots,X_k$
and by $\theta_j$ the probability of inclusion of the single covariate $X_j$.
The prior probability of model $m_i$ is thus:
\begin{equation}\label{eq:inidBernoulli}
 P(m_i) = \prod_{X_j \in \mathcal{X}_i}\theta_j \prod_{X_j \not \in \mathcal{X}_i} (1-\theta_j)
\end{equation}
where we recall that  $\mathcal{X}_i$ is the set of covariates included in model $m_i$.
We call this prior NB, which stands for \textbf{N}on-identical \textbf{B}ernoulli.
The NB prior generalizes the IB prior, retaining its independence assumption but removing
the constraint of the prior probability of inclusion being equal for all covariates.

\section{Credal Model Averaging (CMA)}\label{sec:cma-algo}
CMA generalizes BMA  by substituting the prior over the models by a set of priors  over the models.
The set of priors is called \textit{credal set} \citep{zaffalon2001a}.
We discuss  two different versions of CMA:  CMA$_{ib}$ and CMA$_{nb}$.
CMA$_{ib}$ \citep{2013:23:isipta} generalizes  BMA induced under the IB prior; CMA$_{nb}$ generalizes BMA induced under the NB prior.

\subsection{CMA$_{ib}$}\label{sec:cma-ib}

We start by presenting CMA$_{ib}$.
The BMA induced under the IB prior requires specifying a single value of $\theta$;
instead CMA$_{ib}$  allows $\theta$ to vary within the interval $[\underline{\theta},\overline{\theta}]$.

The constraints $\underline{\theta}$ \textgreater 0 and $\overline{\theta}$ \textless 1 apply to the credal set of CMA$_{ib}$.
For instance, the IB prior with $\theta$=0 assigns zero prior probability to each model
apart from the \textit{null} model which includes no covariates.
The problem is that  such sharp zero probabilities do not change after having seen the data:
prior and posterior probabilities of the models remain identical.
In other words, such prior prevents learning from data.
In the same way the IB prior with $\theta=1$ prevents learning from data.
The IB priors with $\theta=0$  and $\theta=1$ are thus excluded from the credal set.

CMA$_{ib}$  represents a condition \textit{close} to Walley's prior {near-ignorance} \cite[Chap.5.3.2]{walley1991statistical} if one sets $\underline{\theta} =\epsilon$
and $\overline{\theta} =1-\epsilon$. This is the approach followed in \citep{2013:23:isipta}. 

The inferences of CMA$_{ib}$ return
\textit{intervals} of probability rather than a single probability.
For instance CMA$_{ib}$ computes an \textit{interval} for the posterior probability of each class.
The interval shows the sensitivity of the posterior probability on the prior over the models.
Thus, CMA$_{ib}$ automates sensitivity analysis.

The \textit{lower} posterior probability of  class $c_1$ is computed 
as follows:
\begin{eqnarray}\label{eq:min-prob}
 \underline{P}(c_1|D,\mathbf{x})= \min_{\theta\in[\underline{\theta},\overline{\theta}]}\sum_{m_i \in \cal{M}}P(c_1|D,\mathbf{x},m_i)P(m_i|D)=\nonumber\\
 \min_{\theta\in[\underline{\theta},\overline{\theta}]}\sum_{m_i \in \cal{M}}P(c_1|D,\mathbf{x},m_i) \frac{P(D|m_{i})P(m_{i})}{\sum_{m_j \in \cal{M}}P(D|m_{j})P(m_{j})}= \nonumber\\
\min_{\theta\in[\underline{\theta},\overline{\theta}]}\sum_{m_i \in \cal{M}}P(c_1|D,\mathbf{x},m_i)\frac{P(D|m_{i})\theta^{k_i} (1-\theta)^{k-k_i}}{\sum_{m_j \in \cal{M}}P(D|m_{j})\theta^{k_j} (1-\theta)^{k-k_j}}
\end{eqnarray}
where the  marginal likelihoods $P(D|m_i)$ are computed using the BIC approximation of Eqn.(\ref{eb:bic-lik}).
The upper probability of $c_1$ is obtained by  maximizing rather than minimizing expression (\ref{eq:min-prob}).

Since our problem has only two classes, the upper and lower posterior probability of $c_0$ are
readily obtained as:
\begin{eqnarray*}\label{eq:complement}
 \underline{P}(c_0|D,\mathbf{x})=1-\overline{P}(c_1|D,\mathbf{x})\\
 \overline{P}(c_0|D,\mathbf{x})=1-\underline{P}(c_1|D,\mathbf{x})
\end{eqnarray*}

Another relevant inference is the posterior probability of inclusion of a covariate. For instance,
the \textit{lower} probability of  inclusion of covariate $X_j$ is:
\begin{eqnarray}\label{eq:min-prob-inclusion}
 \underline{P}(\beta_j\neq0)= \min_{\theta\in[\underline{\theta},\overline{\theta}]}\sum_{m_i \in \cal{M}}\rho_{ij}P(m_i|D)=\nonumber\\
 \min_{\theta\in[\underline{\theta},\overline{\theta}]}\sum_{m_i \in \cal{M}}\rho_{ij} \frac{P(D|m_{i})P(m_{i})}{\sum_{m_j \in \cal{M}}P(D|m_{j})P(m_{j})}= \nonumber\\
\min_{\theta\in[\underline{\theta},\overline{\theta}]}\sum_{m_i \in \cal{M}}\rho_{ij}\frac{P(D|m_{i})\theta^{k_i} (1-\theta)^{k-k_i}}{\sum_{m_j \in \cal{M}}P(D|m_{j})\theta^{k_j} (1-\theta)^{k-k_j}}
\end{eqnarray}
where the binary variable $\rho_{ij}$ is 1 if model $m_i$ includes covariate $X_j$ and $0$ otherwise.

All such optimization problems  are solved by the analytical procedures reported in the Appendix.
\subsection{CMA$_{nb}$}\label{sec:cma-nb}
CMA$_{nb}$ generalizes BMA induced under the NB prior.
As described in Sec.~\ref{sec:nb-prior}, the NB prior allows specifying a different
prior probability for each covariate.
CMA$_{nb}$ permits also to specify a different \textit{upper} and \textit{lower}
prior probability of inclusion for each covariate.
We denote by
$\underline{\theta}_j$ and $\overline{\theta}_j$
the upper and lower prior probability
of $X_j$.
Moreover, we denote by $\boldsymbol{\underline{\theta}}$ and $\boldsymbol{\overline{\theta}}$
the vectors collecting the upper and lower probabilities of all covariates.
As in the case of CMA$_{ib}$, the probability of inclusion cannot be exactly  zero or one.
A condition close to prior ignorance can be modeled by setting for all covariates
$\underline{\theta}_j=\epsilon$ and
$\overline{\theta}_j=1-\epsilon$.

Let us denote by $K(\boldsymbol{\theta})$ the credal set which contains the admissible values for $\boldsymbol{\theta}$.
The credal set $K(\boldsymbol{\theta})$ is largely different between
CMA$_{nb}$ and CMA$_{ib}$.
Consider a case with three covariates, in which we want to model a condition of ignorance.
Using $\epsilon=0.05$,  under CMA$_{ib}$ we would set $\underline{\theta}$=0.05 and $\overline{\theta}$=0.95.
The credal set $K(\boldsymbol{\theta})$ of CMA$_{ib}$ would have two extreme points: $\{0.05;0.05;0.05\}$; $\{0.95;0.95;0.95\}$.
The credal set $K(\boldsymbol{\theta})$ of CMA$_{nb}$ would have 
$2^3$=8 extreme points: the two extreme points of CMA$_{ib}$ and 6 further ones,
such as for instance $\{0.05;0.95;0.05\}$, $\{0.95;0.95;0.05\}$ and so on.

The choice $\epsilon=0.05$ is a compromise between the objective of representing prior ignorance while not getting too close to 0 and 1. 
The function $f(\theta)$ which represents how the posterior probability of inclusion varies as a function of $\theta$  is continuous. It 
takes value 0 for $\theta$=0 and value 1 for $\theta$=1. Thus, it usually has
large curvature near $\theta$=0 and $\theta$=1.
Very small values of epsilon  would return large CMA intervals, even if the posterior varies narrowly in most of the interval. 

The upper and lower probabilities are computed by solving an optimization in a \textit{k}-dimensional space.
The lower posterior probability of $c_1$ is:
\begin{eqnarray}\label{eq:min-prob-NB}
\underline{P}(c_1|D,\mathbf{x})= \min_{\boldsymbol{\theta}\in[\boldsymbol{\underline{\theta}},\boldsymbol{\overline{\theta}}]}\sum_{m_i \in \cal{M}}P(c_1|D,\mathbf{x},m_i)P(m_i|D)=\nonumber\\
 \min_{\boldsymbol{\theta}\in[\boldsymbol{\underline{\theta}},\boldsymbol{\overline{\theta}}]}\sum_{m_i \in \cal{M}}P(c_1|D,\mathbf{x},m_i) \frac{P(D|m_{i})P(m_{i})}{\sum_{m_j \in \cal{M}}P(D|m_{j})P(m_{j})}= \nonumber\\
\min_{\boldsymbol{\theta}\in[\boldsymbol{\underline{\theta}},\boldsymbol{\overline{\theta}}]}\sum_{m_i \in \cal{M}}P(c_1|D,\mathbf{x},m_i)\frac{P(D|m_{i})\prod_{X_j \in \mathcal{X}_i}\theta_j \prod_{X_j \not \in \mathcal{X}_i} (1-\theta_j)}{\sum_{m_j \in \cal{M}}P(D|m_{j})\prod_{X_j \in \mathcal{X}_i}\theta_j \prod_{X_j \not \in \mathcal{X}_i} (1-\theta_j)}
\end{eqnarray}
Also for CMA$_{nb}$ the upper and lower probability of $c_0$ are the complement to 1 of the
lower and upper probability of $c_1$.

The lower posterior probability of inclusion of covariate $X_j$ is:
\begin{eqnarray}\label{eq:min-postprob-NB}
\underline{P}(\beta_j \neq 0)=
\min_{\boldsymbol{\theta}\in[\boldsymbol{\underline{\theta}},\boldsymbol{\overline{\theta}}]}\sum_{m_i \in \cal{M}}\rho_{ij}P(m_i|D)=\nonumber\\
 \min_{\boldsymbol{\theta}\in[\boldsymbol{\underline{\theta}},\boldsymbol{\overline{\theta}}]}\sum_{m_i \in \cal{M}}\rho_{ij} \frac{P(D|m_{i})P(m_{i})}{\sum_{m_j \in \cal{M}}P(D|m_{j})P(m_{j})}= \nonumber\\
\min_{\boldsymbol{\theta}\in[\boldsymbol{\underline{\theta}},\boldsymbol{\overline{\theta}}]}\sum_{m_i \in \cal{M}}\rho_{ij}\frac{P(D|m_{i})\prod_{X_j \in \mathcal{X}_i}\theta_j \prod_{X_j \not \in \mathcal{X}_i} (1-\theta_j)}{\sum_{m_j \in \cal{M}}P(D|m_{j})\prod_{j \in \mathcal{X}_i}\theta_j \prod_{X_j \not \in \mathcal{X}_i} (1-\theta_j)}
\end{eqnarray}
where  $\rho_{ij}$ is set to 1 if the covariate $X_j$ is included in the model $m_i$ and to 0 otherwise.

The optimization problems of CMA$_{nb}$  cannot be solved analytically; we thus rely on numerical optimization. In particular 
we adopt a local solver provided by the \textbf{NLopt}
software (\url{http://ab-initio.mit.edu/nlopt}).
We use the \textbf{nloptr} package\footnote{\url{http://cran.r-project.org/web/packages/nloptr/index.html}} as 
interface between \textbf{R} and \textbf{NLopt}.
We compute the gradient of the objective function through the symbolic solver of \textbf{R} and then we  provide it to the solver.

\subsection{Sampling the model space}
CMA has been described so far assuming to exhaustively explore 
the model space.
However, data sets with large number of covariates prevent this approach.
In this case it is necessary to sample the model space. Strategies suitable to sample the model space are discussed for instance in
\citep{Clyde2004,boulle2007compression}. 
The CMA algorithms can easily accommodate a set of sampled models.
Denote the set of sampled models by $\cal{M}'$. The CMA inferences could be performed using
the formulas given in Sections \ref{sec:cma-ib} and \ref{sec:cma-nb}, 
provided that the whole model space 
$\cal{M}$ is substituted by  the sampled model space $\cal{M}'$ when summing over the models. 

\subsection{Taking decisions}
Two criteria are commonly used for classification under imprecise probability:
\textit{interval-dominance} and \textit{maximality} \citep{troffaes2007decision}.

According to interval-dominance, class $c_1$ dominates $c_2$ (given covariates $\mathbf{x}$) if:
\begin{equation}
 \underline{P}(c_1|D,\mathbf{x},\theta) > \overline{P}(c_2|D,\mathbf{x},\theta)
\end{equation}

According to maximality, class $c_1$ dominates $c_2$ iff:
\begin{equation}
 P(c_1|D,\mathbf{x},\theta) > P(c_2|D,\mathbf{x},\theta) \,\, \forall \,\, \theta \in [\underline{\theta},\overline{\theta}]\label{eq:maximality}
\end{equation}

If a class is interval-dominant it also maximal \citep{troffaes2007decision}, but not vice versa.
Thus interval-dominance generally return more cautious classifications (more output classes) than
maximality.
Yet, if the class variable is binary the two criteria are equivalent. This  is proven
by the following lemma.

\begin{lemma}
 If the class variable is binary, maximality implies interval-dominance.
\end{lemma}
\begin{proof}
 For a binary class variable:
 \[
 P(c_1|D,\mathbf{x},\theta) = 1-P(c_0|D,\mathbf{x},\theta)  
                              \]
Plugging this expression in Eqn.(\ref{eq:maximality}), we get:
\[
 P(c_1|D,\mathbf{x},\theta) > 1/2   \,\, \forall \,\, \theta \in [\underline{\theta},\overline{\theta}]
 \]
which implies:
\[
 P(c_0|D,\mathbf{x},\theta) < 1/2   \,\, \forall \,\, \theta \in [\underline{\theta},\overline{\theta}]
 \]
Thus, 
\begin{eqnarray}
  \underline{P}(c_1|D,\mathbf{x},\theta) > 1/2\\
  \overline{P}(c_0|D,\mathbf{x},\theta) < 1/2
  \end{eqnarray}
so that $\underline{P}(c_1|D,\mathbf{x},\theta) > \overline{P}(c_0|D,\mathbf{x},\theta)$.
\qed
\end{proof}
Thus when dealing with a binary class (like in our case study), maximality and interval-dominance are equivalent.
For instance,  CMA returns $c_1$ as a prediction  if both its upper \textit{and} lower posterior probability are \textit{greater} than 1/2.
In this case the instance is \textit{safe}: the most probable class does not vary with the prior over the models.

Instead, CMA returns \textit{the set of classes} $\{c_0,c_1\}$ if the posterior probability intervals of the two classes overlap.
This happens if both classes have upper probability \textit{greater} than 1/2 and lower 
probability \textit{smaller} than 1/2.
In this case the instance is \textit{prior-dependent}: one class or the other is more probable depending on the prior over the models.

A final consideration regards the case in which the prior used to induce BMA is included in the credal set of CMA.
In this case the posterior probability computed by BMA is included within the posterior interval computed by CMA.
When CMA returns a single class, BMA and CMA predictions match.
\section{Case study}\label{sec:case-study}

Data regarding the distribution of Alpine marmot (\emph{marmota marmota}) burrows were collected by AM and other collaborators in the summer of 2010 and 2011, in an Alpine valley in  Northern Italy.
To develop the species distribution model we divide the explored area into cells of 10 x 10m, obtaining a data set of 9429 cells.
The fraction of presence (\textit{prevalence}) is 436/9429= 0.046.


Considering that the Alpine marmot prefers south-facing slopes ranging between 1600 and 3000 m a.s.l. \citep{Lopez2010},
we introduce \textit{altitude} and \textit{slope} as covariates.
A third relevant piece of information is the aspect, namely the angle between the maximum gradient of the terrain and the North.
We represent the aspect by  introducing two covariates (\emph{northitude} and \emph{eastitude}),
corresponding respectively to the cosine and the sine of the aspect.
Northitude  and the eastitude are proxies for the amount and the temporal distribution of sunlight received during the day.
The fifth covariate is the \textit{curvature}, which measures the upward convexity (or concavity) of the terrain.
The sixth and last covariate is the \textit{soil cover}, namely the proportion of terrain not covered by vegetation. We obtain the soil cover from a digital map of the land use\footnote{The database, known as DUSAF2.0, was retrieved at: \url{http://www.cartografia.regione.lombardia.it/geoportale}.}.

The Alpine marmot is a mobile species, which uses a huge territory for its activities.
Therefore the decision of establishing a burrow depends also on the conditions of the surrounding cells.
For this reason we average the value of each covariate over a circular \emph{buffer area} of 2 ha around the cell being analyzed.

\subsection{Interviewing experts}
We asked three experts for the prior probability of inclusion of each covariate; the results are reported in  Table \ref{tab:experts}.
The pool of experts is composed by  two scientists who published several papers on the species
(Dr. Bernat Claramunt L\'{o}pez and Prof. Walter Arnold) and a master student (Mrs. Viviana Brambilla) who participated to the collection of marmot data without analyzing them.
The prior beliefs of the experts are shown in Table \ref{tab:experts}.
The labels of first, second and third expert are randomly assigned
to hide whose are the prior beliefs.

The first expert provided us with a single probability value for each covariate, while the two other experts provided us with interval probabilities.
The third expert provides intervals strongly skewed either towards inclusion  or exclusion.

\begin{table}[htbp]
\centering
 \begin{tabular}{ccccccccc}\toprule
   \textit{}	& 	\multicolumn{3}{c}{\textbf{Experts}}	&	& & & \multicolumn{2}{c}{\textbf{Priors}}\\
& First Expert	&	Second Expert	&	Third Expert &	& & &CMA & BMA\\	
&&	&	&	& & &(convex hull) & (central point)\\	\midrule
\textit{altitude}	&	0.95	&	[0.80-0.95]	&	[0.90-0.95]	& & &	&[0.80-0.95]	 & 0.87\\
\textit{slope }		&	0.50	&	[0.70-0.95]	&	[0.05-0.10]	& &	& &[0.05-0.95] 	& 0.50	\\
\textit{curvature} 	&	0.40	&	[0.40-0.60]	&	[0.05-0.10]	 & && &	[0.05-0.60]	 & 0.27\\
\textit{northitude} 	&	0.60	&	[0.60-0.80]	&	[0.90-0.95]	& &&	 &[0.60-0.95]	 & 0.77 \\
\textit{eastitude} 	&	0.60	&	[0.60-0.90]	&	[0.05-0.10]	& & &&	[0.05-0.90]	 & 0.50\\
\textit{soil cover}	&	0.95	&	[0.70-0.95]	&	[0.90-0.95]	& & &&	[0.70-0.95]	& 0.82\\
\bottomrule
 \end{tabular}
\caption{Probability of inclusion according to the three experts;
imprecise model of prior knowledge (convex hull);
precise model of prior knowledge (central point of the convex hull).}\label{tab:experts}
\end{table}

We aggregate in two different ways the expert beliefs.
Firstly we take their \textit{convex hull} in the spirit of imprecise probability.
We will later use such convex hulls to represent (imprecise) prior knowledge within
CMA$_{nb}$, which allows for different specification of the lower and upper prior probability of inclusion
of each covariate.
Secondly we take the \textit{central point } of the convex hull in the more traditional spirit
of representing prior knowledge by a single prior distribution.
We will later use such information to design a NB prior for  BMA.

The difference among such two approaches can be readily appreciated.
Consider  slope, for which the experts have strongly different opinions.
The convex hull of its prior probability of inclusion is a \textit{wide} interval (0.05--0.95), which
appropriately represents a condition of substantial \textit{ignorance}.	
The central point approach yields prior probability of inclusion 0.5, which represents 
prior \textit{indifference} about the inclusion/exclusion of  the covariate.
As pointed out by \cite[Chap.5.5]{walley1991statistical},  a model of prior \textit{indifference} is inappropriate to model 
the substantial uncertainty which instead characterizes a state of \textit{ignorance}.
\section{Results}\label{sec:results}
We induce  BMA under three priors: IB ($\theta$=0.5, non-informative);
BB ($\alpha=\beta=1$, non-informative); NB (informative).
We call these three models BMA$_{ib}$, BMA$_{bb}$ and BMA$_{nb}$.
For BMA$_{nb}$ we set the prior probability of inclusion of each covariate
equal to the central point of the convex hull of the expert beliefs reported in
Table \ref{tab:experts}. Thus BMA$_{nb}$ embodies domain knowledge.

We also consider three variants of CMA.
The first is CMA$_{ib}$ with $\overline{\theta}=0.95$ and $\underline{\theta}=0.05$.
This is the model originally proposed in \citet{2013:23:isipta} and represents
a condition close to prior-ignorance, though  under the restrictive assumption
of the prior probability of all covariates being equal.
The second is CMA$_{nb}$
with the prior-ignorant configuration $\underline{\theta_j}=0.05$ and $\overline{\theta_j}=0.95$ for each covariate.
As already discussed,
the credal set of CMA$_{nb}$ contains a much wider variety of priors compared to CMA$_{ib}$;
thus we expect CMA$_{nb}$ to be much more imprecise than CMA$_{ib}$.

The third model is CMA$_{exp}$.
This is a variant of CMA$_{nb}$ which embodies partial prior knowledge: upper and lower probability of inclusion of each covariate correspond
to the upper and lower bound of the convex hull of the expert beliefs (Table~\ref{tab:experts}). CMA$_{exp}$ consider narrower prior interval of inclusion for the covariates than CMA$_{nb}$ and
thus it should be more determinate than CMA$_{nb}$.
This application shows how  CMA can be easily tuned to represent prior ignorance or prior knowledge.

\subsection{Posterior probability of inclusion of covariates}

We induce the three BMAs and CMAs using the \textit{whole} data set (9429 instances).
Table \ref{tab:postprobinc} reports the posterior probability of inclusion of each covariate under the different models.
Such posterior is a point estimate for the BMAs and an interval estimate for the CMAs.
We recall that the beta-binomial prior is \textit{not} included in the credal
set of the CMAs: for this reason its estimate can lie outside of the CMA intervals.

\begin{table}[htbp]
\centering
\begin{tabular}{ccccccc}\toprule
Covariate	& \multicolumn{3}{c}{\textbf{BMA}}	& \multicolumn{3}{c}{\textbf{CMA}}\\
     	&BMA$_{ib}$	&BMA$_{nb}$	& BMA$_{bb}$ &CMA$_{ib}$	&CMA$_{nb}$	&CMA$_{exp}$\\ \midrule
altitude	&1.00		&1.00		&1.00		&[1.00 - 1.00]	&[1.00 - 1.00]	&[1.00 - 1.00]\\
slope    	&1.00		&1.00		&1.00		&[1.00 - 1.00]	&[1.00 - 1.00]	&[1.00 - 1.00]\\
curvature	&0.02		&0.02		&0.01		&[0.00 - 0.27]	&[0.00 - 0.39]	&[0.00 - 0.03]\\
northitude 	&1.00		&1.00		&1.00		&[1.00 - 1.00]	&[1.00 - 1.00]	&[1.00 - 1.00]\\
eastitude	&1.00		&1.00		&1.00		&[1.00 - 1.00]	&[1.00 - 1.00]	&[1.00 - 1.00]\\
soil cover	&0.97		&0.99		&0.94		&[0.66 - 1.00]	&[0.55 - 1.00]	&[0.99 - 1.00]\\ \bottomrule
\end{tabular}
\caption{Posterior probability of inclusion of each covariates estimated by  different models.}
\label{tab:postprobinc}
\end{table}

The most important variables are altitude, slope, eastitude and northitude, whose posterior probability of inclusion is estimated as 1 by all the considered model.
In particular the posterior probability of inclusion of such covariates
is not sensitive on the prior over the models, also thanks to the huge data set.
Remarkably in this case  the  posterior intervals of CMA
collapse into a single point, lower and upper posterior probability of such covariates being both one.

The  results for curvature are less unanimous.
The BMAs recognize it  as irrelevant, estimating a posterior probability not larger than 0.02.
Yet, the two CMAs induced under prior-ignorance (CMA$_{ib}$ and CMA$_{nb}$) achieve a   much less certain conclusions,  estimating the upper posterior probability of inclusion as 0.3 or 0.4.
This hardly allows to safely discard such covariate.
Interestingly, CMA$_{exp}$ achieves a much sharper conclusion, assigning to the curvature an upper posterior probability of inclusion of only 0.03, in line with the Bayesian models.
Thus CMA$_{exp}$ achieves (on this large data set) conclusions which are as sharp as those of the Bayesian models, but much safer as it
does not commit to a single prior.

The soil cover is recognized as relevant by the BMAs, its posterior probability
of inclusion ranging between 0.94 and 0.99 depending on the prior over the models.
Yet, according to CMA$_{ib}$ and CMA$_{nb}$ its \textit{lower} posterior probability of inclusion does not exceed 0.7. Also in this case CMA$_{exp}$ achieves a much sharper conclusion, assigning to soil cover a posterior probability comprised between 0.99 and 1, further showing the beneficial effect of prior knowledge.

The results presented so far are obtained using the entire dataset for training the models.
It is however interesting repeating the analysis with smaller training sets,
in which the choice of the prior over the models is likely to have a greater effect.
We thus down-sampled the data set, creating data sets of size comprised between
$n$=30 and $n$=6000.
The training sets are \textit{stratified}: they contain the same proportion of presence of the original data set.
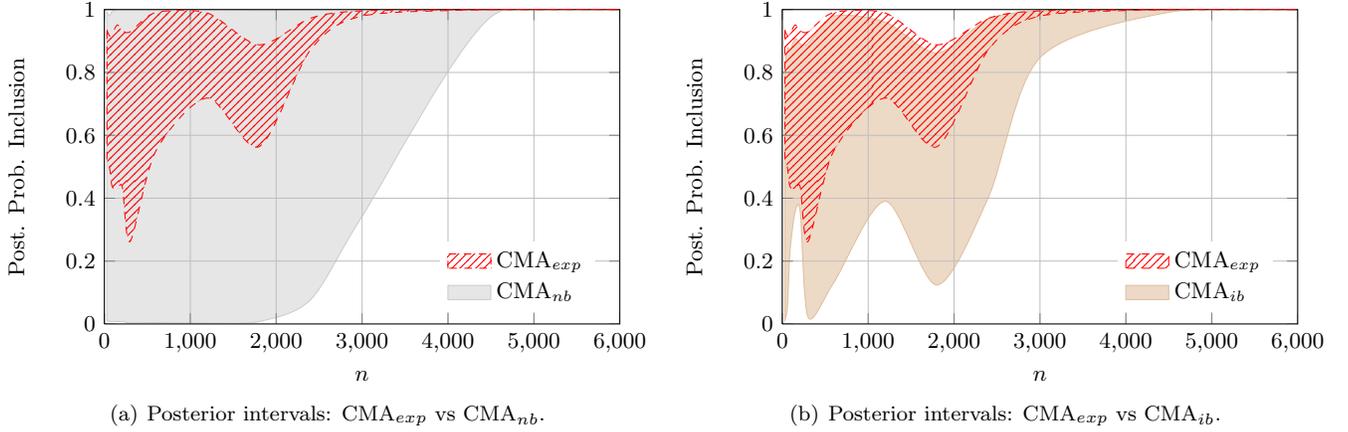
\begin{figure}[!ht]
\subfigure[Posterior intervals: CMA$_{exp}$ vs CMA$_{nb}$.]
{
\begin{tikzpicture}[transform shape,scale=0.95]

  \begin{axis}[xlabel=$n$, ylabel=Post. Prob. Inclusion,
  legend style={draw=none, at={(0.65,0.1)},anchor=west},  grid=major,
  xmin=0,xmax=6000, height= 0.25 * \textheight, width= 0.5 * \textwidth,
  every axis/.append style={font=\small},ymax=1,ymin=0,axis on top]

  \addplot[draw=none,smooth,forget plot,stack plots=y] 
  table [col sep=comma, x=ntrain, y=Ignorant_min] {altitude_postprob.csv} \closedcycle;

  \addplot[color = gray!40 ,fill=gray!20,smooth,area legend,stack plots=y] 
  table [col sep=comma, x=ntrain, y expr={\thisrow{Ignorant_max}-\thisrow{Ignorant_min}}] {altitude_postprob.csv} \closedcycle;
  \legend{CMA$_{nb}$}

  \end{axis}
  
  \begin{axis}[xlabel=$n$, ylabel=Post. Prob. Inclusion,
  legend style={draw=none, at={(0.65,0.2)},anchor=west},  grid=major,
  xmin=0,xmax=6000, height= 0.25 * \textheight, width= 0.5 * \textwidth,
  every axis/.append style={font=\small},ymax=1,ymin=0,axis on top,hide axis]

  \addplot[draw=none,smooth,forget plot,stack plots=y]
  table [col sep=comma, x=ntrain, y=Expert_min] {altitude_postprob.csv} \closedcycle;

  \addplot[dashed,color=red,pattern=north east lines,pattern color = red,smooth, area legend,stack plots=y]
  table [col sep=comma, x=ntrain, y expr ={\thisrow{Expert_max}-\thisrow{Expert_min}}] {altitude_postprob.csv} \closedcycle;
  
  \legend{ CMA$_{exp}$}

  \end{axis}

  \end{tikzpicture}
}
\subfigure[Posterior intervals: CMA$_{exp}$ vs CMA$_{ib}$.]{
  \begin{tikzpicture}[transform shape,scale=0.95]
  
    \begin{axis}[xlabel=$n$, ylabel=Post. Prob. Inclusion,
  legend style={draw=none, at={(0.65,0.1)},anchor=west},  grid=major,
  xmin=0,xmax=6000, height= 0.25 * \textheight, width= 0.5 * \textwidth,
  every axis/.append style={font=\small},ymax=1,ymin=0,axis on top]

  \addplot[draw=none,smooth,forget plot,stack plots=y] 
  table [col sep=comma, x=ntrain, y=Fixed_min] {altitude_postprob.csv} \closedcycle;

  \addplot[color = brown!50 ,fill=brown!30,smooth,area legend,stack plots=y] 
  table [col sep=comma, x=ntrain, y expr={\thisrow{Fixed_max}-\thisrow{Fixed_min}}] {altitude_postprob.csv} \closedcycle;
  
 \legend{CMA$_{ib}$}
  
  \end{axis}
  
  \begin{axis}[xlabel=$n$, ylabel=Post. Prob. Inclusion,
  legend style={draw=none, at={(0.65,0.2)},anchor=west},  grid=major,
  xmin=0,xmax=6000, height= 0.25 * \textheight, width= 0.5 * \textwidth,
  every axis/.append style={font=\small},ymax=1,ymin=0,axis on top,hide axis]

  \addplot[draw=none,smooth,forget plot,stack plots=y]
  table [col sep=comma, x=ntrain, y=Expert_min] {altitude_postprob.csv} \closedcycle;

  \addplot[dashed,color=red,pattern=north east lines,pattern color = red,smooth, area legend,stack plots=y]
  table [col sep=comma, x=ntrain, y expr ={\thisrow{Expert_max}-\thisrow{Expert_min}}] {altitude_postprob.csv} \closedcycle;
  
  \legend{ CMA$_{exp}$}

  \end{axis}
  
  \end{tikzpicture}
}
%
%
%
%
%
\caption{Upper and lower posterior probability of inclusion of altitude
for different CMA and BMA models.}
\label{fig:PIP}
\end{figure}

In Fig.\ref{fig:PIP}(a) and \ref{fig:PIP}(b) we show as an example how the upper and lower posterior probability of inclusion
of the altitude covariate varies with $n$.
We show the upper and lower probability of inclusion computed by different CMAs.
The gap between upper and lower probability of inclusion narrows down as the sample size increases, eventually converging towards a punctual probability.
The gap between upper and lower probability computed by
CMA$_{exp}$ is generally narrower than those of
CMA$_{nb}$   and CMA$_{ib}$.
This is the beneficial effect of expert knowledge.
The curve are non-monotonic, probably because we performed just once the whole procedure.
Averaging over many repetitions would yield smoother curves.

\subsection{Comparing  CMA and BMA predictions}\label{sec:expe-protocol}

We consider training sets with dimension comprised between
30 and 1500. Beyond this size no significant changes are detected.

For each sample size we repeat 30 times the procedure of
i) building  a training set by randomly down-sampling the original data set;
ii) training the different BMAs and CMAs;
iii) assessing the model  predictions on the test set, constituted by 1000 instances not included in the training set.
Training and test sets are \textit{stratified}: they have the same \textit{prevalence} (fraction of presence data) of the original data set.

The most common indicator of performance for classifiers is the \textit{accuracy}, namely the proportion of instances correctly classified using 0.5 as probability threshold.
Accuracy ranges between 0.93 and 0.97 depending on the sample size. The  accuracies of the different BMAs are pretty close.
However, a skewed distribution of the classes can misleadingly inflate the value of accuracy. If for instance a species is absent from 90\% of the sites, a trivial classifier
which always returning absence would achieve 90\% accuracy without providing any information.
The AUC (area under the receiver-operating curve) \citep{fawcett2006introduction} is  more robust than accuracy, being insensitive to class unbalance.
The AUC of a random guesser is 0.5; the AUC of a perfect predictor is 1.

\begin{figure}[!ht]
 \centering
\begin{subfigure} 
 \centering
\begin{tikzpicture}[]
\begin{axis}[xlabel=$n$, ylabel=AUC, legend style={draw=none, legend pos=south east}, grid=major, xmin=0,xmax=600, height= 0.2 * \textheight, width= 0.45 * \textwidth,
every axis/.append style={font=\footnotesize}]
\addplot[smooth,mark=*] table [col sep=comma, x=sample_size, y=Uniform_AUC] {accuracy_utility.csv};
\addplot [smooth, thick, dashed,red,every mark/.append style={fill=red!80!black},mark=square*]
table [col sep=comma, x=sample_size, y=BB_AUC] {accuracy_utility.csv};
\addplot [smooth,thick,mark=triangle]
table [col sep=comma, x=sample_size, y=Central_AUC] {accuracy_utility.csv};
\legend{BMA$_{ib}$,BMA$_{bb}$,BMA$_{nb}$}
\end{axis}
\end{tikzpicture}
\end{subfigure}
 \begin{subfigure}
 \centering
%
\begin{tikzpicture}[]
\begin{axis}[xlabel=$n$, ylabel=Recall, legend style={draw=none, legend pos=north east}, grid=major, xmin=0,xmax=600, height= 0.2 * \textheight, width= 0.45 * \textwidth,
every axis/.append style={font=\footnotesize},ytick={0.05,0.10,0.15},yticklabels={0.05,0.10,0.15}] 
\addplot[smooth,mark=*] table [col sep=comma, x=sample_size, y=Uniform_sensitivity] {accuracy_utility.csv};
\addplot [smooth, thick, dashed,red,every mark/.append style={fill=red!80!black},mark=square*]
table [col sep=comma, x=sample_size, y=BB_sensitivity] {accuracy_utility.csv};
\addplot [smooth,thick,mark=triangle]
table [col sep=comma, x=sample_size, y=Central_sensitivity] {accuracy_utility.csv};
\legend{BMA$_{ib}$,BMA$_{bb}$,BMA$_{nb}$}
\end{axis}
\end{tikzpicture}
\end{subfigure}

\caption{AUC and recall of BMA under different priors.
The plots are truncated at
$n$=600 since after this value no further change is observed.  Each point represents the average over 30 experiments.}\label{fig:acc-bma}
\end{figure}
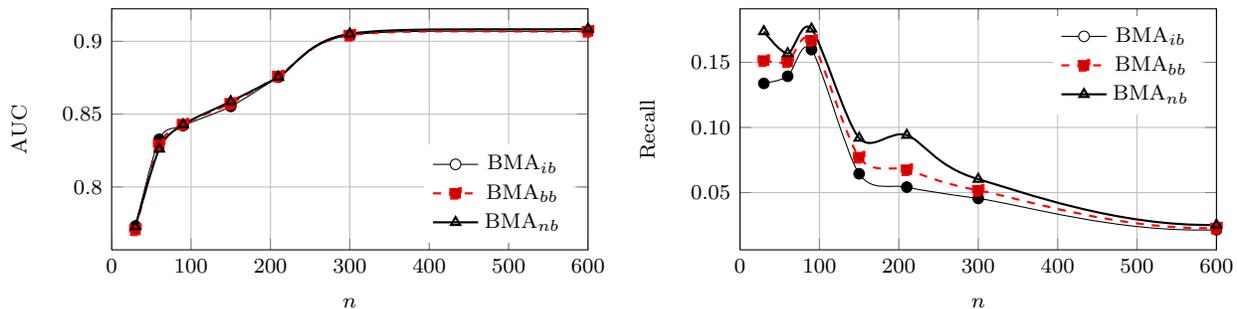

Figure \ref{fig:acc-bma} shows the AUC of BMA using different priors over the models.
The plots are truncated at $n$=600, since no further significant changes are observed going beyond
this amount of data.
Larger training sets allow to better learn the model and
result in larger AUC values. The
impact of the prior on AUC is quite thin.
Overall, BMA performs well, its AUC being generally superior to 0.8.

In this case study, the probability of presence is much lower than the probability of absence.
Another meaningful indicator of performance is thus the recall (percentage of existing burrows whose presence is correctly predicted). 
The recall of BMA$_{nb}$ is consistently higher than that of the other BMAs (Figure \ref{fig:acc-bma}b), thus benefiting from expert knowledge.
Indicators such as precision and recall are used when the problem is cost-sensitive.

\begin{figure}[!htbp]
 \centering
 \includegraphics[width=6cm]{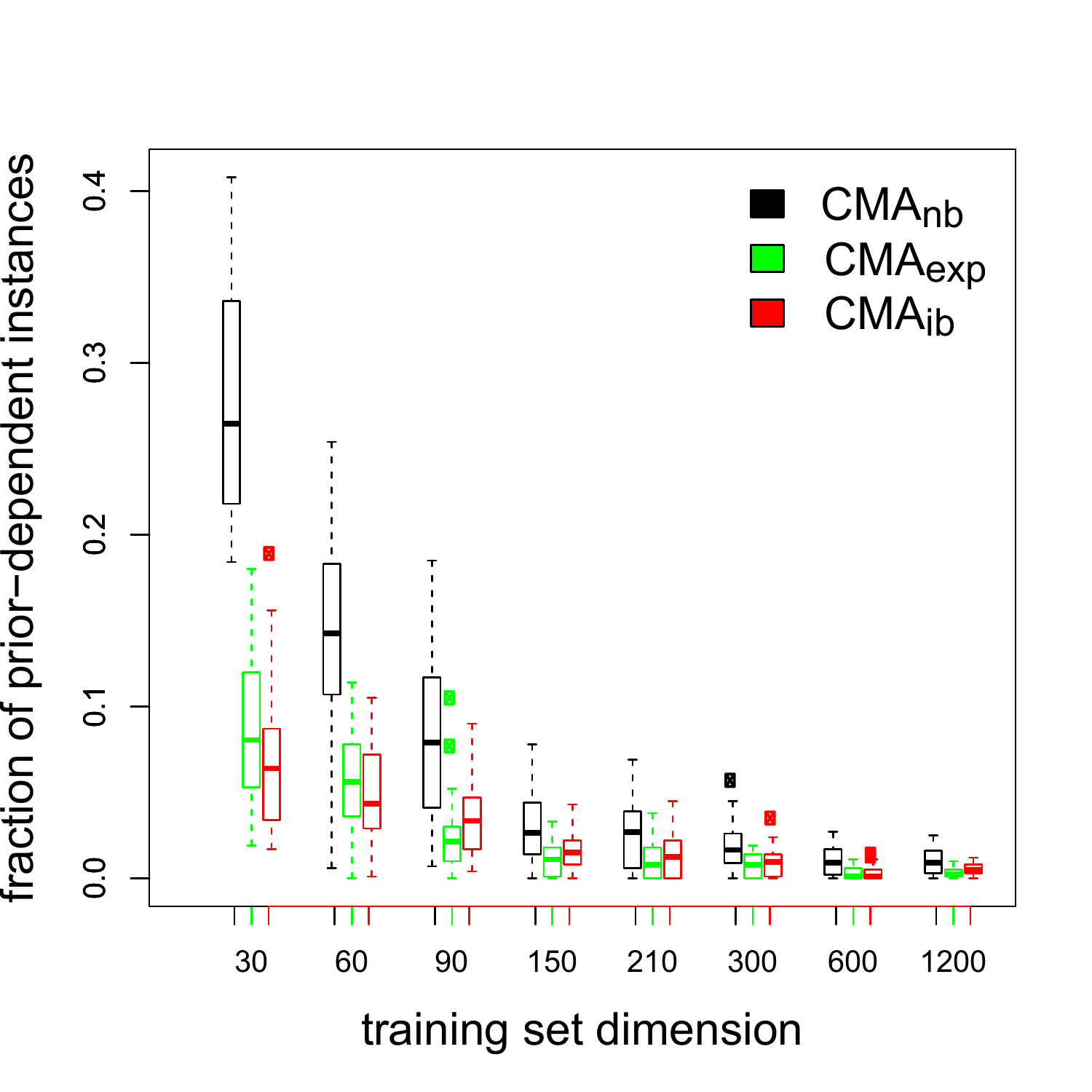}
\caption{Fraction of prior-dependent instances for the different CMA variants.
For each CMA variant and each training set dimension, a box and whisker plot is reported.
The limit of the box represent the interquartile range, while the thick line inside the box is the median.
The limits of the whisker extend out of the box to the most extreme data point which is no more than 1.5 times the interquartile range from the box.
Points that lay out of the whiskers are shown as circles.
\label{fig:marm-indeterminate}}
\end{figure}

We now analyze the CMA results.
We call \textit{indeterminate} classifications the cases in which CMA suspends the judgment
returning both classes.
The  percentage of indeterminate classifications (\textit{indeterminacy}) of the different
CMAs is shown in Fig.\ref{fig:marm-indeterminate}.
The  indeterminacy consistently decreases with the sample size.
This is the well-known behavior of credal classifiers, which become more determinate
as more data are available.
CMA$_{nb}$ is the most indeterminate algorithm;
CMA$_{ib}$ is the least indeterminate algorithm.
The reason of this behavior lies in the different definition of the credal sets:
while both algorithms aim at representing a condition close to prior-ignorance,
the credal set of CMA$_{nb}$ contains a much wider set of priors and results in higher imprecision.
Interestingly, CMA$_{exp}$
is \textit{much} less indeterminate than CMA$_{nb}$ thanks to prior knowledge.

We recall that  any CMA  algorithm divides the instances into two groups: the \textit{safe} and the \textit{prior-dependent} ones. CMA returns a single class on the safe instances and
both classes on the prior-dependent ones.
We therefore separately assess the accuracy  of BMA on the safe and on the prior-dependent instances. This analysis is more meaningful when the prior used to induce
BMA is included in the credal set of CMA. We thus consider the following pairs:
BMA$_{ib}$ vs CMA$_{ib}$; BMA$_{ib}$ vs CMA$_{nb}$; BMA$_{ib}$ vs CMA$_{exp}$.

\begin{figure}[!htb]
 \centering
\subfigure[BMA$_{ib}$ vs CMA$_{ib}$.]{\includegraphics[width=6cm]{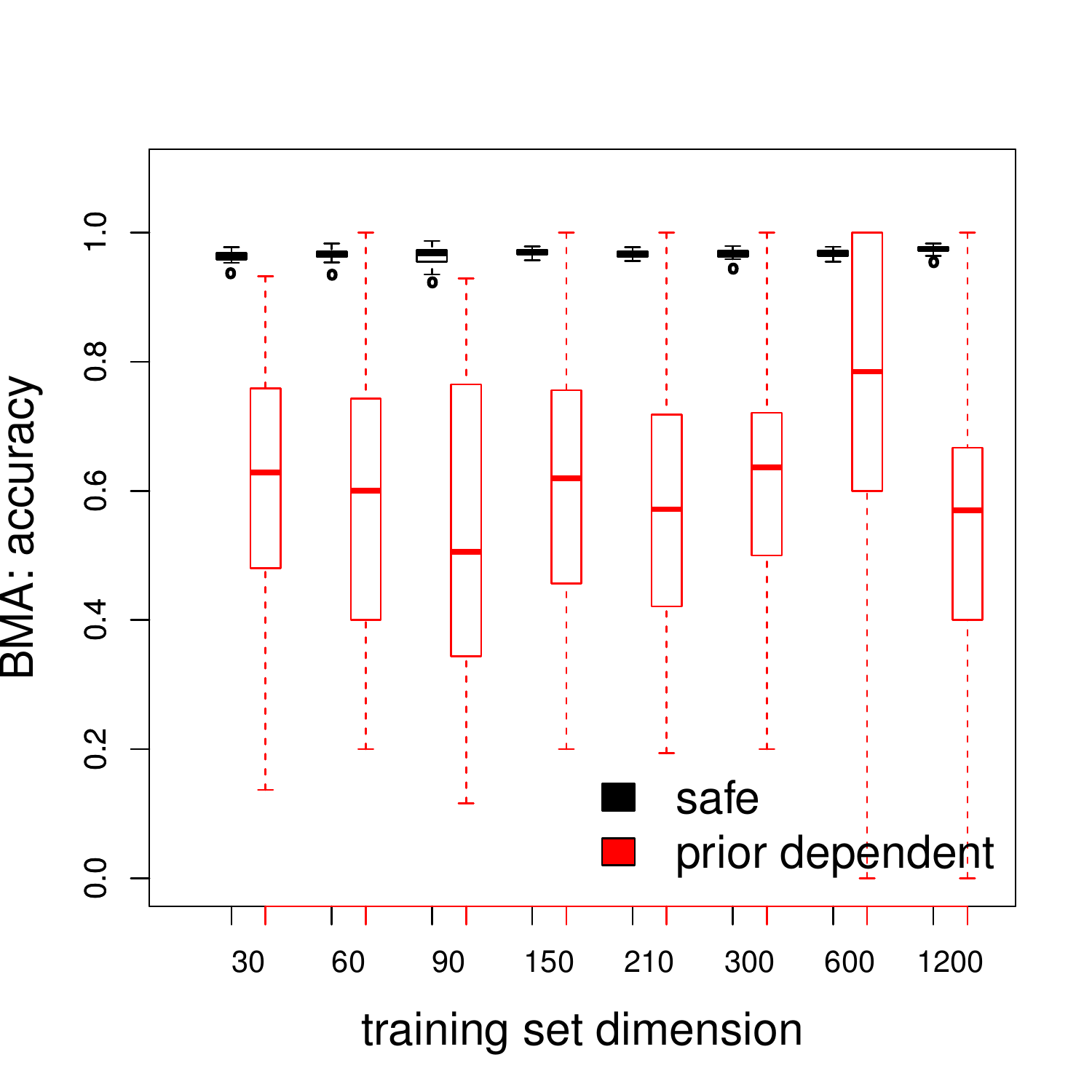}}
\subfigure[BMA$_{nb}$ vs CMA$_{nb}$.]{\includegraphics[width=6cm]{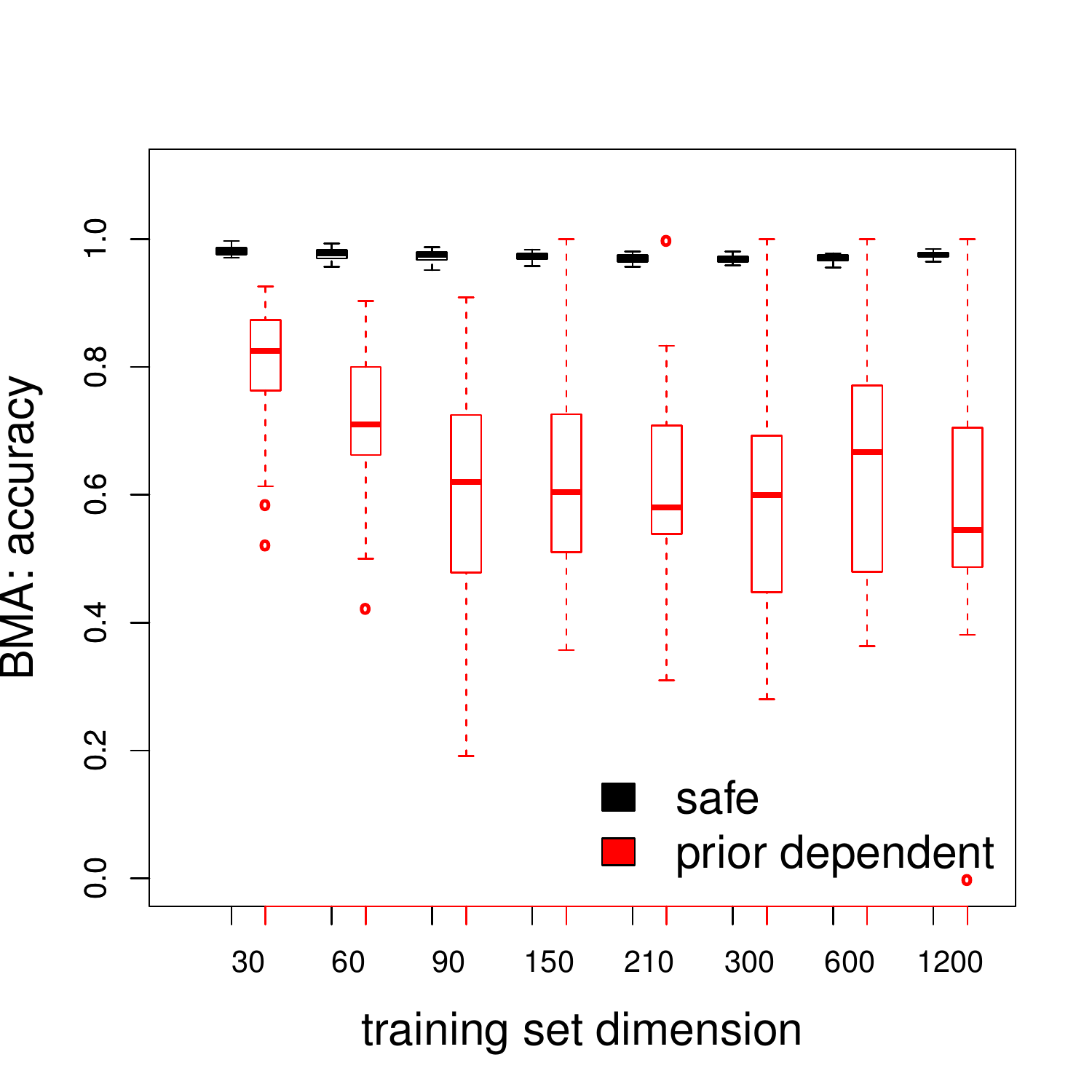}}
\subfigure[BMA$_{nb}$ vs CMA$_{exp}$.]{\includegraphics[width=6cm]{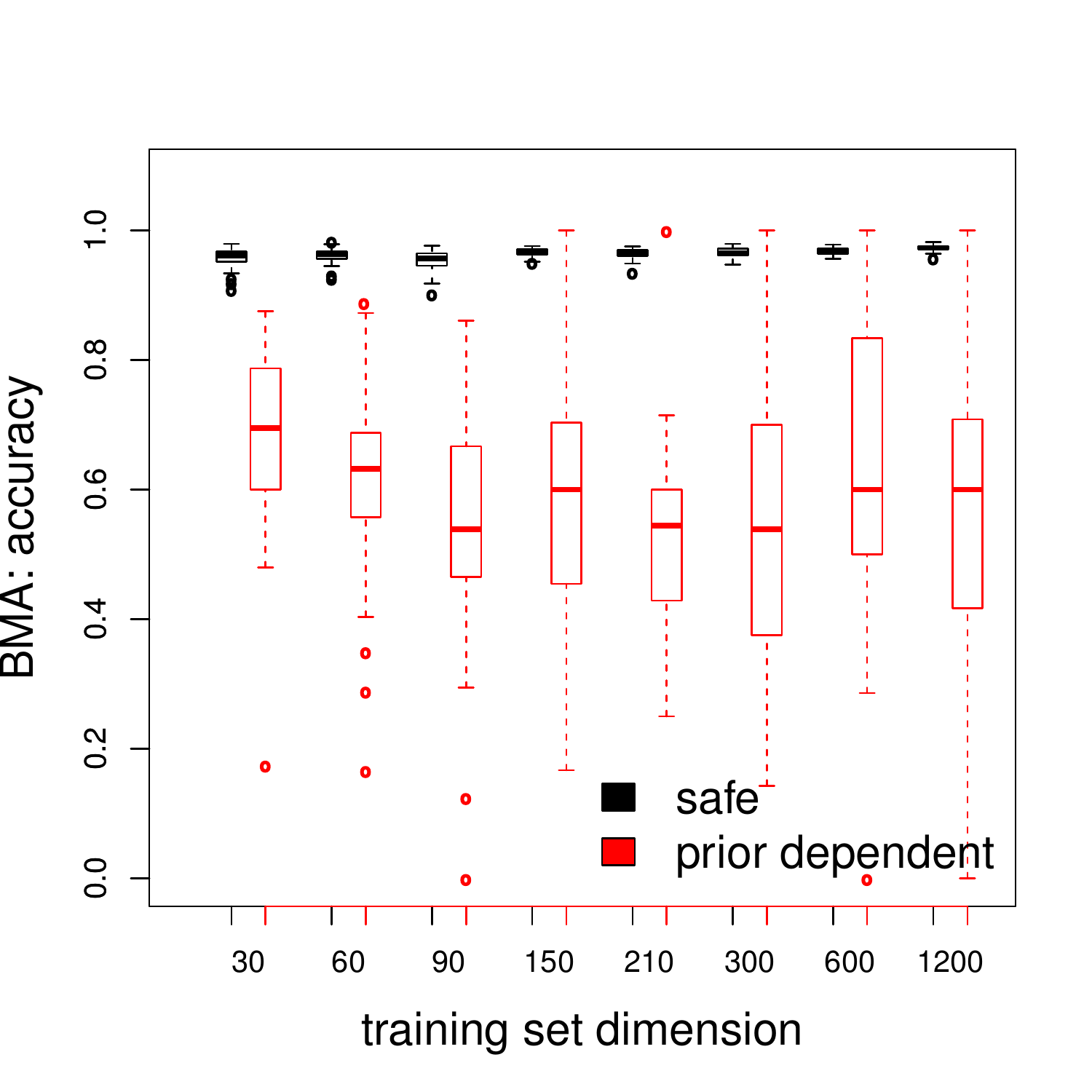}}
\caption{The accuracy of BMA sharply drops on the prior-dependent instances recognized by CMA.
Each boxplot refers to 30 experiments. The boxplots are computed as described in the caption of Figure \ref{fig:marm-indeterminate}.
The corresponding fraction of indeterminate predictions is shown in Figure \ref{fig:marm-indeterminate}.
\label{fig:acc-free}}
\end{figure}

Figure \ref{fig:acc-free} (a) compares the accuracy
of BMA$_{ib}$ on the instances recognized as safe and prior-dependent by CMA$_{ib}$.
On the prior-dependent instances  the accuracy of BMA$_{ib}$ severely drops, getting almost close to random guessing. On a data set with two classes, a random  guesser achieves accuracy 0.5; the average accuracy of BMA$_{ib}$ on the prior-dependent instances is 0.6. On the safe instances, the accuracy of BMA is  above 90\%.
CMA$_{ib}$ thus uncovers a small yet non-negligible set of instances (between 2\% and 8\%) over which BMA$_{ib}$ performs poorly because of prior-dependence.
The phenomenon is already known in the literature of the credal classification \citep{corani2008learning,corani2008credal}.
It has been moreover observed \citep{corani2008credal,2013:23:isipta} that such
doubtful instances are hardly identifiable by looking
at the BMA posterior probabilities.
Detecting prior-dependence using BMA
would require cross-checking the predictions of many BMAs, each induced with a different prior over the models.
This is quite unpractical and is not usually done.


In Figure \ref{fig:acc-free}(b) we compare the accuracy of BMA$_{nb}$ on the instances recognized as safe and as prior-dependent by CMA$_{nb}$.  The prior of BMA$_{ib}$ is contained in the credal set of
CMA$_{nb}$.
The results is qualitatively similar to the previous one, with a sharp drop of accuracy of 
BMA$_{ib}$ on the instances recognized as prior-dependent by CMA$_{ib}$.
Yet, the accuracy of BMA on the prior-independent instances is higher (about 70\%) compared
to the previous case, since  CMA$_{nb}$ is \textit{much} more indeterminate than CMA$_{ib}$.

Eventually, we compare the accuracy of BMA$_{nb}$ on the instances recognized as safe and as prior-dependent by CMA$_{exp}$. Note that the prior of BMA$_{nb}$  is included in the credal set of CMA$_{exp}$.
On average, the accuracy of BMA$_{nb}$ on the prior-dependent instances is about 60\%. The situation is quite similar to the comparison of BMA$_{ib}$ and CMA$_{ib}$.

\subsection{Utility measures}
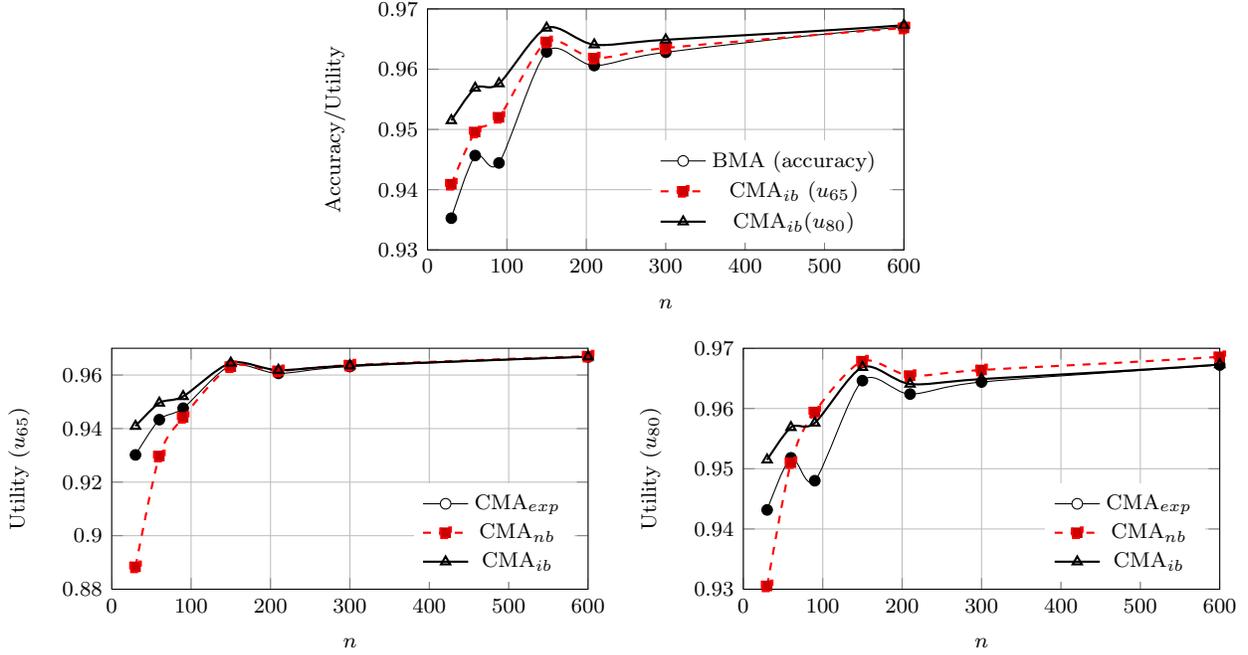
\begin{figure}[!ht]
 \centering
\begin{subfigure}
 \centering
\begin{tikzpicture}[]
\begin{axis}[xlabel=$n$, ylabel=Accuracy/Utility, legend style={draw=none, legend pos=south east}, grid=major, xmin=0,xmax=600, every axis/.append style={font=\footnotesize}, height= 0.2 * \textheight, width= 0.45 * \textwidth,ymin=0.93,ymax=0.97]
\addplot[smooth,mark=*] table [col sep=comma, x=sample_size, y=Uniform_accuracy] {accuracy_utility.csv};
\addplot [smooth, thick, dashed,red,every mark/.append style={fill=red!80!black},mark=square*]
table [col sep=comma, x=sample_size, y=Fixed_util65] {accuracy_utility.csv};
\addplot [smooth,thick,mark=triangle]
table [col sep=comma, x=sample_size, y=Fixed_util80] {accuracy_utility.csv};
\legend{BMA (accuracy), CMA$_{ib}$ ($u_{65}$), CMA$_{ib}$($u_{80}$)}
\end{axis}
\end{tikzpicture}
\end{subfigure}

\begin{subfigure}
 \centering
\begin{tikzpicture}[]
\begin{axis}[xlabel=$n$, ylabel=Utility ($u_{65}$), legend style={draw=none, legend pos=south east}, grid=major, xmin=0,xmax=600,
every axis/.append style={font=\footnotesize}, height= 0.2 * \textheight, width= 0.45 * \textwidth,ymin=0.88,ymax=0.97]
\addplot[smooth,mark=*] table [col sep=comma, x=sample_size, y=Expert_util65] {accuracy_utility.csv};
\addplot [smooth, thick, dashed,red,every mark/.append style={fill=red!80!black},mark=square*]
table [col sep=comma, x=sample_size, y=Ignorance_util65] {accuracy_utility.csv};
\addplot [smooth,thick,mark=triangle]
table [col sep=comma, x=sample_size, y=Fixed_util65] {accuracy_utility.csv};
\legend{CMA$_{exp}$, CMA$_{nb}$, CMA$_{ib}$}
\end{axis}
\end{tikzpicture}
\end{subfigure}
\begin{subfigure}
 \centering
\begin{tikzpicture}[]
\begin{axis}[xlabel=$n$, ylabel=Utility ($u_{80}$), legend style={draw=none, legend pos=south east}, grid=major, xmin=0,xmax=600,
every axis/.append style={font=\footnotesize}, height= 0.2 * \textheight, width= 0.45 * \textwidth,ymin=0.93,ymax=0.97]
\addplot[smooth,mark=*] table [col sep=comma, x=sample_size, y=Expert_util80] {accuracy_utility.csv};
\addplot [smooth, thick, dashed,red,every mark/.append style={fill=red!80!black},mark=square*]
table [col sep=comma, x=sample_size, y=Ignorance_util80] {accuracy_utility.csv};
\addplot [smooth,thick,mark=triangle]
table [col sep=comma, x=sample_size, y=Fixed_util80] {accuracy_utility.csv};
\legend{CMA$_{exp}$, CMA$_{nb}$, CMA$_{ib}$}
\end{axis}
\end{tikzpicture}
\end{subfigure}
\caption{Accuracies of BMA$_{ib}$ versus utility of CMA$_{ib}$. The plot is truncated at
$n$=600 since after this value no further change of accuracies is observed.}
\label{fig:util}
\end{figure}

To further compare the classifiers we adopt
the utility measures introduced in \cite{Zaffalon20121282}, which we briefly describe in the following.
The starting point is the \emph{discounted accuracy}, which  rewards a prediction containing $m$ classes with $1/m$ if it contains the true class and with 0 otherwise.
Within a betting framework based on fairly general assumptions, discounted-accuracy is the only score which satisfies some fundamental properties for assessing both determinate and indeterminate classifications.
In fact, for a determinate classification (a single class is returned)
discounted-accuracy corresponds to the traditional classification accuracy.
Yet discounted-accuracy has severe shortcomings.
Consider two medical doctors, doctor \textit{random} and doctor \textit{vacuous}, who should diagnose  whether a patient is \emph{healthy} or \emph{diseased}.  Doctor \textit{random} issues uniformly random diagnosis; doctor \textit{vacuous} instead always returns both
categories, thus admitting to be ignorant.  Let us assume that the hospital profits a quantity of money  proportional to the discounted-accuracy achieved by its doctors at each visit.  Both doctors have the same \textit{expected} discounted-accuracy for each visit,  namely $1/2$.  For the hospital, both doctors provide the same \textit{expected} profit from each visit, but with a substantial difference: the profit of doctor vacuous has no variance.  Any risk-averse hospital manager should thus prefer doctor vacuous over doctor random: under risk-aversion, the expected utility increases with expectation of the rewards  and decreases with their variance \citep{levy1979approximating}.
To model this fact, it is necessary to apply a utility function to the discounted-accuracy score assigned to each instance.
We designed the utility function according to \cite{Zaffalon20121282}:
the utility of a correct and determinate classification (discounted-accuracy 1) is 1;
the utility of a wrong classification (discounted-accuracy  0) is 0; 
the utility of an accurate but indeterminate classification consisting of two classes (discounted-accuracy 0.5)  is assumed to lie between 0.65 and 0.8.
Notice that, following the first two rules, the utility of a traditional classifier corresponds to its accuracy.
Two quadratic utility functions are derived, passing respectively  through $\{u(0)=0,u(0.5)=0.65,u(1)=1\}$ and  $\{u(0)=0,u(0.5)=0.8,u(1)=1\}$, denoted as $u_{65}$ and $u_{80}$ respectively.
Utility of credal classifiers and accuracy of determinate classifiers can be directly compared.

Figure \ref{fig:util}(a) compares the accuracy of BMA$_{ib}$
with the utility of CMA$_{ib}$, considering both $u_{65}$ and $u_{80}$ as
utility functions.
In both cases the utility of CMA$_{ib}$ is higher than the accuracy of BMA$_{ib}$;
the extension to imprecise probability proves valuable.
The gap is narrower under $u_{65}$  and larger under $u_{80}$, as the latter function
assigns higher value to the indeterminate classifications.
Moreover, the gap gets thinner as the sample size increases: as the data set grows large,
CMA$_{ib}$ becomes less indeterminate and thus closer to BMA$_{ib}$.

In Figure \ref{fig:util}(b) and \ref{fig:util}(c)  we
compare the different CMAs  using
$u_{65}$ and $u_{80}$.
According to $u_{65}$, the best performing model is CMA$_{ib}$, followed
by CMA$_{exp}$ and by CMA$_{nb}$.
The function $u_{65}$ assign  a limited value to the indeterminate classifications.
Thus under this utility the most  determinate algorithm (CMA$_{ib}$) achieves the highest score;
the least determinate (CMA$_{nb}$) achieves the lowest score.

The same situation is found under $u_{80}$, but only for small sample sizes ($n$ \textless 60).
For larger $n$,  CMA$_{nb}$ becomes the highest scoring CMA.
The point is that  CMA$_{nb}$ is the most imprecise model, and under $u_{80}$ the imprecision is highly rewarded.
Depending thus on the considered utility function, a different variant of CMA achieves the best performance.

These results are fully reasonable: each CMA provides a different trade-off between informativeness and robustness. Moreover, the two utility function represents two quite different types of risk aversion.
It can be expected that the they differently rank  the various CMAs.
Yet, it is somehow puzzling that CMA$_{exp}$ is never ranked as the top CMA, despite being the only algorithm which provides both a flexible model of prior and a robust elicitation of prior knowledge.

A partial explanation is that  the utility measures are derived assuming all the errors to be equally costly.
In a problem like ours missing a presence is likely to be much costlier than missing an absence.
Yet, there is currently no way to assess credal classifiers assuming 
unequal misclassification costs. 


\section{Conclusions}
BMA is the state of the art approach to deal with model uncertainty.
Yet, the results of the BMA analysis can well depend
on the prior which has been set over the models,
especially on small data sets.

To robustly deal with this problem, CMA adopts a \textit{set} of priors over the models rather than a single prior.
CMA automates sensitivity analysis and  detects  prior-dependent instances, on which BMA  is almost random guessing.
To identify the prior-dependent instances without using CMA, one would need to cross-check the predictions of many BMAs, each induced with a different prior over the models.
This would be very unpractical.

We have presented three different versions of CMA.
They represent different types of ignorance or partial knowledge a priori.
Experiments show that extending BMA to imprecise probability is indeed valuable.

However, deciding which variant of CMA performs better
is not easy, partially because the trade-off between robustness
and informativeness is a subjective matter and partially because
there are currently no score for assessing credal classifiers when the cost of the misclassification
errors are unequal.

An interesting avenue for future works is to develop CMA algorithms for the analysis of prior-data conflict.
This approach would allow for detecting major discrepancies between prior distribution and data, 
thus checking automatically the soundness of the opinion of the experts.
A recent proposal for prior-data conflict in the context of credal classification is discussed by \citep{Masegosa2013}.

\section*{Acknowledgments}
We are grateful to Dr. Bernat Claramunt L\'{o}pez (Center for Ecological Research and Forestry Applications, Unit of Ecology of the Autonomous University of Barcelona),
Prof. Walter Arnold (University of Veterinary Medicine in Vienna) and Mrs. Viviana Brambilla (Master student at the Universidade Tecnica de Lisboa)
who provided us with their prior probability of inclusion of the covariates.
The research in this paper has been partially supported by the Swiss NSF grants no. 200020-132252.
The work has been performed during Andrea Mignatti's PhD, supported by Fondazione Lombardia per l'Ambiente (project SHARE- Stelvio).
We are moreover grateful to the anonymous reviewers.

\appendix
\section*{Appendix A: solution of the CMA optimization problems.}
We show in the following how to solve the optimization problems (minimization and  maximization)
for IB-CMA.

Let us define the $k$  sets $\mathcal{M}_1 \ldots  \mathcal{M}_k$ which include all the models containing respectively $\{1,2,\ldots,k\}$ covariates.
For instance, $\mathcal{M}_2$ contains all the models which include two covariates.
The models included in the same set have the same prior probability; for instance the  prior probability of a model belonging to the set
$\mathcal{M}_j$ is $\theta^{j} (1-\theta)^{k-j}$.
We denote $L_{j}=\sum_{m_i \in \mathcal{M}_{j}} P(D|m_i)$.

The definition of variable $Z_{j}$ depends instead on the problem being addressed, as detailed in the following table:

\begin{center}
\begin{tabular}{lcl} \toprule
\textit{Problem} & \textit{Equation number} & Definition of $Z_j$\\ \midrule
Lower/Upper prob. of presence & (\ref{eq:min-prob}) & $ \sum_{m_i \in \mathcal{M}_{j}} P(c_1|D,\mathbf{x}) P(D|m_i)$\\
Lower/Upper prob. of inclusion of covariate $X_j$& (\ref{eq:min-prob-inclusion}) & $ \sum_{m_i \in \mathcal{M}_{j}} \rho_{ij} P(D|m_i)$\\ \bottomrule
\end{tabular}
\end{center}
where the binary variable $\rho_{ij}$ is 1 if model $m_i$ includes the covariate $X_j$ and 0 otherwise.

The function to be optimized (minimized or maximized) can be written as:
\begin{equation} \label{eq:pred_binom}
 h(\theta) := \dfrac {\sum_{j=0} ^{k} \theta^{j} (1-\theta)^{k-j} Z_{j}}{\sum_{j=0} ^{k} \theta^{j} (1-\theta)^{k-j}  L_{j} }
\end{equation}

In the interval [$\underline{\theta},\overline{\theta}$], the maximum and minimum of $h(\theta)$
should lie either in the boundary points $\theta=\overline{\theta}$ and $\theta=\underline{\theta}$, or in an internal point of the interval
in which the first derivative of $h(\theta)$ is 0.
Let us introduce $f(\theta)= \sum_{j=0} ^{k} \theta^{j} (1-\theta)^{k-j}  Z_{j}$ and $g(\theta)= \sum_{j=0} ^{k} \theta^{j} (1-\theta)^{k-j}  L_{j}$.
The first derivative $h'(\theta)$ is:

\begin{equation} \label{eq:solution}
  h'(\theta) = \dfrac{f'(\theta)g(\theta)-f(\theta)g'(\theta)}{g(\theta)^2},
\end{equation}

where $g(\theta)$ is strictly positive because $L_{j}$ is a sum of marginal likelihoods. We can  therefore search the solutions looking only at the numerator
$f'(\theta)g(\theta)-f(\theta)g'(\theta)$,
which is a polynomial of degree $k(k-1)$ and thus has $k(k-1)$ solutions in the complex plain.
We are interested only in the $real$ solutions that lie in the interval ($\underline{\theta},\overline{\theta}$).
Such solutions, together with the boundary solutions $\theta=\overline{\theta}$ and $\theta=\underline{\theta}$, constitute the set of \emph{candidate solutions}.
To find the minimum and the maximum  $h(\theta)$, we evaluate $h(\theta)$
in each candidate solution point, and eventually we retain the minimum or the maximum among such values.

\section*{Appendix B: the beta-binomial prior for Bayesian model averaging.}
The Beta-binomial (BB) prior is discussed for instance by \citep[Chap.3.2]{bernardo2009bayesian}. 
It  treats  parameter $\theta$ as a random variable with Beta prior distribution:
$\theta \sim Beta(\alpha, \beta)$.
The prior probability of model $m_i$ which includes $k_i$ covariates is obtained by marginalizing out the Beta distribution:
\begin{eqnarray*}
 P(m_i)= \int_0^1 \theta^{k_i} (1-\theta)^{k-k_i} p(\theta) d\theta=\nonumber \\
\int_0^1 \theta^{k_i} (1-\theta)^{k-k_i} \frac{\theta^{\alpha-1} (1-\theta)^{\beta-1}}{B(\alpha,\beta)} d\theta= \nonumber \\
\frac{\Gamma(\alpha+\beta)}{\Gamma(\alpha)\Gamma(\beta)}
\int_0^1 \theta^{\alpha+k_i-1} (1-\theta)^{\beta+k-k_i-1} d\theta=\nonumber \\
\frac{\Gamma(\alpha+\beta)}{\Gamma(\alpha)\Gamma(\beta)}
\frac{\Gamma(\alpha+k_i)\Gamma(\beta-k_i+k)}{\Gamma(\alpha+\beta+k)} \label{eq:int-beta}
\end{eqnarray*}
where the last passage leading to  Eqn.\ref{eq:int-beta} is explained considering that the Beta distribution integrates to 1:
\begin{eqnarray*}
\frac{\Gamma(\alpha+\beta+k)}{\Gamma(\alpha+k_i)\Gamma(\beta-k_i+k)}\int_0^1 \theta^{\alpha+k_i-1} (1-\theta)^{\beta+k-k_i-1} d\theta=1 
\Longrightarrow \\
\int_0^1 \theta^{\alpha+k_i-1} (1-\theta)^{\beta+k-k_i-1} d\theta = \frac{\Gamma(\alpha+k_i)\Gamma(\beta-k_i+k)}{\Gamma(\alpha+\beta+k)}
\end{eqnarray*}
Under the choice $\alpha=\beta=1$,  the Beta distribution becomes  \textit{uniform} and
the probability of model $m_i$ which contains $k_i$ covariates becomes:
\begin{equation*}
P(m_i)=\frac{\Gamma(\alpha+\beta)}{\Gamma(\alpha)\Gamma(\beta)}
\frac{\Gamma(\alpha+k_i)\Gamma(\beta-k_i+k)}{\Gamma(\alpha+\beta+k)}=
\frac{\Gamma(2)}{\Gamma(1)\Gamma(1)}
\frac{\Gamma(1+k_i)\Gamma(1-k_i+k)}{\Gamma(2+k)}=\frac{k_i!k-k_i!}{k+1!} 
\end{equation*}
This gives the prior probability of a model with $k_i$ covariates.
The probability of the model size $W$ to be equal to $k_i$ is obtained by 
combining Eqn.(\ref{eq:prob-beta-single}) with the observation that
that there are $k \choose k_i$ models which contain $k_i$ covariates:
\begin{equation*}
P(W=k_i)= P(m_i)\cdot{k \choose k_i} = \frac{k_i!k-k_i!}{k+1!} \frac{k!}{k-k_i!k_i!}=\frac{1}{k+1}
\,\ \forall m_i
\end{equation*}
The model size is thus uniformly distributed, as a result of having set a uniform prior on $\theta$.
\bibliography{biblio}
\bibliographystyle{model2-names}
\end{document}